\documentclass[12pt]{article}
\usepackage{amsmath}
\usepackage{graphicx,psfrag,epsf}
\usepackage{tabularx,supertabular,subfigure,amsmath,latexsym,amssymb}
\usepackage{float,epsfig,multirow,rotating,times}
\usepackage{upgreek,wrapfig}
\usepackage{enumerate}
\usepackage[authoryear]{natbib}
\usepackage{comment}

\newcommand{\blind}{0}

\newcommand {\ctn}{\citet}

\renewcommand{\a}{\ensuremath{\alpha}}

\renewcommand{\d}{\ensuremath{\delta}}
\newcommand{\e}{\ensuremath{\epsilon}}

\newcommand{\bm}{\mathbf}
\renewcommand{\l}{\lambda}
\newcommand {\bzero}{\mbox{\boldmath $0$}}
\newcommand {\bxi}{\mbox{\boldmath $\xi$}}
\newcommand {\bdelta}{\mbox{\boldmath $\delta$}}
\newcommand {\bmu}{\mbox{\boldmath $\mu$}}
\newcommand {\bSigma}{\mbox{\boldmath $\Sigma$}}
\newcommand{\dint}{\int\displaylimits}
\newcommand{\boldm}{\mathbf m}

\newcommand{\bI}{\mathbf I}
\newcommand{\bM}{\mathbf M}

\newcommand{\bp}{\mathbf p}
\newcommand{\bq}{\mathbf q}
\newcommand{\bv}{\mathbf v}
\newcommand{\bw}{\mathbf w}
\newcommand{\bx}{\mathbf x}
\newcommand{\by}{\mathbf y}
\newcommand{\bz}{\mathbf z}
\newcommand{\be}{\pmb\e}

\newtheorem{theorem}{Theorem}
\newtheorem{proof}{Proof}
\newtheorem{definition}{Definition}
\newcommand{\topline}{\hrule height 1pt width \textwidth \vspace*{2pt}}
\newcommand{\botline}{\vspace*{2pt}\hrule height 1pt width \textwidth \vspace*{4pt}}
\newtheorem{algo}{Algorithm} 

\numberwithin{equation}{section}
\numberwithin{algo}{section}
\numberwithin{table}{section}
\numberwithin{figure}{section}

\newcommand{\statesp}{\ensuremath{\mathcal X}}
\newcommand{\Y}{\ensuremath{\mathcal Y}}
\newcommand{\D}{\ensuremath{\mathcal D}}

\newcommand{\imply}{\Longrightarrow}
\newcommand{\R}{\ensuremath{\mathbb R}}

\newcommand{\supr}[2]{{#1}^{(#2)}}

\newcolumntype{Z}{>{\centering\arraybackslash}X}%

\date{}
\begin{document}

\bibliographystyle{natbib}

\def\spacingset#1{\renewcommand{\baselinestretch}%
{#1}\small\normalsize} \spacingset{1}


\if0\blind
{
  \title{\bf Supplement to ``Markov Chain Monte Carlo Based on Deterministic Transformations"}
   \author{Somak Dutta\thanks{Corresponding e-mail: sdutta@galton.uchicago.edu}\\
    Department of Statistics\\
    University of Chicago\\
    and\\
   Sourabh Bhattacharya\\
   Bayesian and Interdisciplinary Research Unit\\
  Indian Statistical Institute
 }
\maketitle
 } \fi

\if1\blind
{
 \bigskip
 \bigskip
 \bigskip
 \begin{center}
 {\LARGE\bf Markov Chain Monte Carlo Based on Deterministic Transformations}
 \end{center}

  \medskip
  } \fi

  \bigskip

\renewcommand\thefigure{S-\arabic{figure}}
\renewcommand\thetable{S-\arabic{table}}
\renewcommand\thesection{S-\arabic{section}}

\spacingset{1.45}  

Throughout, we refer to our main article \ctn{Dutta13} as DB. 

\section{Proof of detailed balance for TMCMC}
\label{sec:detailed_balance}
The detailed balance condition is proved as follows: Suppose $\by = T_{\bz}(\bx,\be) \in T_{\bz}(\bx,\Y)$, 
then $\bx = T_{\bz^c}(\by,\be)$. Hence, the kernel $K$ satisfies,
\begin{eqnarray*}
 \pi(\bx)K(\bx \to \by) & = & \pi(\bx)~P(T_{\bz})~g(\be)\min\left\{1,\dfrac{P(T_{\bz^c}) \pi(\by)}
 {P(T_{\bz})\pi(\bx)}J_{\bz}(\bx,\be)\right\} \\
& = & g(\be)~\min\left\{\pi(\bx)~P(T_{\bz}),\pi(\by)P(T_{\bz^c})~J_{\bz}(\bx,\be)\right\}
\end{eqnarray*}
and
\begin{eqnarray*}
 \pi(\by)K(\by \to \bx) & = & \pi(\by)~P(T_{\bz^c})~g(\be)J_{\bz}(\bx,\be)\min\left\{1,\dfrac{P(T_{\bz})
 \pi(\bx)}{P(T_{\bz^c})\pi(\by)}J_{\bz^c}(\by,\be)~\right\} \\
& = & g(\be)~\min\left\{\pi(\by)~P(T_{\bz^c})J_{\bz}(\bx,\be),\pi(\bx)P(T_{\bz})\right\}
\end{eqnarray*}
where $J_{\bz}(\bx,\be) = \left| \partial (T_{\bz}(\bx,\be),\be)/\partial (\bx,\be)\right|$ satisfies
\[ J_{\bz^c}(T_{\bz}(\bx,\be),\be) \times J_{\bz}(\bx,\be) = 1 \quad\textrm{ since }\quad  
T_{\bz^c}(T_{\bz}(\bx,\be),\be) = \bx.  \]

\section{General TMCMC algorithm based on a single $\e$}
\label{sec:general_tmcmc}

\begin{algo}\label{algo:GTM:multivar2} \topline General TMCMC algorithm based on a single $\e$.
\botline \normalfont \ttfamily
\begin{itemize}
 \item Input: Initial value $\supr{\bm x}{0}$, and number of iterations $N$. 
 \item For $t=0,\ldots,N-1$
\begin{enumerate}
 \item Generate $\e \sim g(\cdot)$ and an index $i \sim \mathcal M(1;p_1,\ldots,p_{3^k-1})$ independently.
 Again, actual simulation from the high-dimensional multinomial distribution is not necessary; 
 see Section 3.1 of DB. 
 \item \[ \bm x' = T_{\bz_i}(\supr{\bm x}{t}, \e) \quad \textrm{ and } 
 \quad \alpha(\supr{\bm x}{t}, \e) = \min\left(1, \dfrac{P(T_{\bz^c_i})}{P(T_{\bz_i})} ~\dfrac{\pi(\bm x')}{\pi(\supr{\bm x}{t})} 
 ~\left|\frac{\partial (T_{\bz_i}(\supr{\bm x}{t}, \e),\e)}{\partial(\supr{\bm x}{t}, \e)}\right| \right)\]
\item Set \[ \supr{\bm x}{t+1}= \left\{\begin{array}{ccc}
 \bm x' & \textsf{ with probability } & \a(\supr{\bm x}{t},\e) \\
 \supr{\bm x}{t}& \textsf{ with probability } & 1 - \a(\supr{\bm x}{t},\e)
\end{array}\right.\]
\end{enumerate}
\item End for
\end{itemize}
\botline \rmfamily
\end{algo}

\section{Convergence properties of additive TMCMC}\label{sec:theorems}
In this section we prove some convergence properties of the TMCMC in the case of the additive transformation. 
%
Before going into our main result we first borrow some definitions from the MCMC literature. 
\begin{definition}[Irreducibility] A Markov transition kernel $K$ is $\varphi-$irreducible, where $\varphi$ is a nontrivial measure, if for every $x \in \statesp$ and for every measurable set $A$ of $\statesp$ with $\varphi(A) > 0$, 
there exists $n\in \mathbb N$, such that $K^n(x,A) > 0.$
\end{definition}
\begin{definition}[Small set] A measurable subset $E$ of $\statesp$ is said to be small if there is an $n \in \mathbb N$, a constant $c > 0$, possibly depending on $E$ and a finite measure $\nu$ such that
\[ K^n(x,A) \geq c~\nu(A), \qquad \forall ~A \in \mathcal{B}(\statesp),~\forall ~x \in E\]
\end{definition}

\begin{definition}[Aperiodicity] A Markov kernel $K$ is said to be periodic with period $d > 0$ if the state-space $\statesp$ can be partitioned into $d$ disjoint subsets $\statesp_1,\statesp_2,\ldots,\statesp_d$ with 
\[ K(x,\statesp_{i+1}) = 1 ~\forall ~x\in \statesp_i, ~i=1,2,\ldots,d-1\]
and $K(x,\statesp_1) = 1~\forall ~x \in \statesp_d$.

A Markov kernel $K$ is aperiodic if for no $d\in\mathbb N$ it is periodic with period $d$.
\end{definition}

\subsection{Additive transformation with singleton $\e$}
Consider now the case where $\statesp = \R^k$, $\D = \R$ and 
$T_{\bz}(\mathbf x,\e) = (x_1+z_1a_1\e,x_2+z_2a_2\e,\ldots,x_k+z_ka_k\e)$ 
where, for $i=1,\ldots,k$, $z_i=\pm 1$, and $a_i>0$. In this case $\Y = [0,\infty)$. 
Suppose that $g$ is a density on $\Y$.

\begin{theorem}
 Suppose that $\pi$ is bounded and positive on every compact subset of $\R^k$ and that $g$ is 
 positive on every compact subset of $(0,\infty)$. Then the chain is $\l$-irreducible, aperiodic. 
 Moreover every nonempty compact subset of $\R^k$ is small.
\end{theorem}

\begin{proof}
Without loss we may assume that $a_i=1;$ $i=1,\ldots,k$. 
For notational convenience we shall prove the theorem for $k = 2$. The general case can be 
seen to hold with suitably defined `rotational' matrices on $\R^k$ similar to \eqref{formula:rotmats1}.

Suppose $E$ is a nonempty compact subset of $\R^k$. Let $C$ be a compact rectangle whose sides are parallel to the diagonals $\{(x,y) ~:~ |y| = |x|\}$ and containing $E$ such that $\l(C) > 0$. We shall show that $E$ is small, i.e., $\exists ~c > 0$ such that 
\[ K^2(\bx,A) \geq c \l_C(A) \qquad \forall A \in \mathcal B(\R^2) \textrm{ and } \forall x \in E.\]
It is clear that the points reachable from $\bx$ \emph{in two steps} are of the form
\[\begin{pmatrix} x_1 \pm \e_1 \pm \e_2 \\ x_2 \pm \e_1 \pm \e_2\end{pmatrix},\qquad \e_1 \geq 0, \e_2 \geq 0 \]
Thus, if we define the matrices
\begin{equation}\label{formula:rotmats1}
\begin{split}
 & M_1 = \begin{pmatrix} 1 & 1 \\ 1 & -1 \end{pmatrix} \quad M_2 = \begin{pmatrix} -1 & 1 \\ 1 & 1 \end{pmatrix} \quad M_3 = \begin{pmatrix} 1 & -1 \\ -1 & -1 \end{pmatrix} \quad M_4 = \begin{pmatrix} -1 & -1 \\ -1 & 1 \end{pmatrix} \\
& \tilde M_1 = \begin{pmatrix} 1 & 1 \\ -1 & 1 \end{pmatrix} \quad \tilde M_2 = \begin{pmatrix} 1 & -1 \\ 1 & 1 \end{pmatrix} \quad \tilde M_3 = \begin{pmatrix} -1 & 1 \\ -1 & -1 \end{pmatrix} \quad \tilde M_4 = \begin{pmatrix} -1 & -1 \\ 1 & -1 \end{pmatrix}
\end{split}
\end{equation}
then the points reachable from $\bx$ \emph{in two steps}, other than the points lying on the diagonals passing through $\bx$ itself, are of the form
\[ \bx + M_i\left(\begin{smallmatrix} \e_1 \\ \e_2\end{smallmatrix}\right) \quad\textrm{ and }\quad  \bx + \tilde M_i\left(\begin{smallmatrix} \e_1 \\ \e_2\end{smallmatrix}\right), \quad \e_1 > 0, \e_2 > 0,~i=1,\ldots,4.\]

Define
\[ m = \inf_{\by \in C}\pi(\by) > 0 \qquad M = \sup_{\by \in C}\pi(\by) < \infty \qquad a = \inf_{0 < \e < R}g(\e) > 0\]
where $R$ is the length of the diagonal of the rectangle $C$\footnote{Actually $R/\sqrt{2}$ suffices.}.
Fix an element $\bx \in E$. For any set $A \in \mathcal B(\R^2)$, let $A^* = A\cap C$ and define,
\begin{equation}
\begin{split}
A_i & = \{ \be \in (0,\infty)^2 ~:~ \bx + M_i \be\in A^* \} \\
\tilde A_i & = \{ \be \in (0,\infty)^2 ~:~ \bx + \tilde M_i \be\in A^* \}
\end{split}
\end{equation}
The need for defining such sets illustrated in the following example: to make a transition from the state $\bx$ to a state in $A^*$ in two steps, first making a forward transition in both coordinates and then a forward transition in first coordinate and a backward transition in the second coordinate is same as applying the transformation $\bx \to \bx + M_1\be$ for some $\be \in A_1$ in two steps, i.e. first 
\[\bx \to \bx + M_1(\e_1,0)^T = \bx + (\e_1,\e_1)^T \quad\textrm{ then } \quad \bx + M_1(\e_1,\e_2)^T \to \bx + M_1\be\]
Also note that for any $\be = (\e_1,\e_2)\in A_i$, $A^* \subset C$ implies that the intermediate point $\bx + M_i(\e_1,0)^T \in C$ and similarly for $\tilde A_i~(i=1,\ldots,4)$.
Now, with ${\underline p}$ and $\bar p$ as the minimum and maximum of the move probabilities. 
\begin{eqnarray}\label{eqn:smallsetadditive}
 & & K^2(\bx,A) \geq K^2(\bx,A^*) \nonumber \\
& \geq & {\underline p}^2\sum_{i=1}^4 \dint_{A_i} g(\e_1)g(\e_2)\min\left\{ \frac{{\underline p}\pi(\bx + M_i(\e_1,0)^T)}{{\bar p}\pi(\bx)},1\right\}\min\left\{ \frac{{\underline p}\pi(\bx + M_i(\e_1,\e_2)^T)}{{\bar p}\pi(\bx + M_i(\e_1,0)^T)},1\right\}d\e_1d\e_2 \nonumber\\
& + & {\underline p}^2\sum_{i=1}^4 \dint_{\tilde A_i} g(\e_1)g(\e_2)\min\left\{ \frac{{\underline p}\pi(\bx + \tilde M_i(\e_1,0)^T)}{{\bar p}\pi(\bx)},1\right\}\min\left\{ \frac{{\underline p}\pi(\bx + \tilde M_i(\e_1,\e_2)^T)}{{\bar p}\pi(\bx + \tilde M_i(\e_1,0)^T)},1\right\}d\e_1d\e_2 \nonumber \\
& \geq & {\underline p}^2 a^2\left(\min\left\{\dfrac{\underline p m}{\bar p M},1\right\}\right)^2\left(\sum_{i=1}^4\l(A_i) + \sum_{i=1}^4\l(\tilde A_i) \right) \nonumber \\
& = & {\underline p}^2 a^2\left(\min\left\{\dfrac{\underline p m}{\bar p M},1\right\}\right)^2 \times 2 \times \sum_{i=1}^4\l(A_i)
\end{eqnarray}
Since $(\e_1,\e_2) \in A_i \iff (\e_2,\e_1) \in \tilde A_i$, so that, $\l(A_i) = \l(\tilde A_i)$.
Now notice that, if we define for $i=1,\ldots,4$
\[ f_i : (0,\infty)^2 \to \R^2 \ni \be \mapsto \bx + M_i\be \]
 and 
\[ A_\bx = \{ (\e,0)^T ~:~ \e > 0, (x_1 \pm \e, x_2 \pm \e) \in A^* \}\]
then,
\[A^* = \bigcup_{i=1}^4f_i(A_i\cup A_x) \quad\imply\quad \l(A^*) = \sum_{i=1}^4f_i(A_i) \quad = \quad 2 \times \sum_{i=1}^4\l(A_i),\]
since, $f_i(A_i)$'s are pairwise disjoint, $\l(f_i(A_x)) = 0$ and $\l(f_i(A_i)) = 2\l(A_i)$ for $1\leq i \leq 4$. It follows from \eqref{eqn:smallsetadditive} that
\[ K^2(\bx,A) \quad\geq \quad {\underline p}^2 a^2\left(\min\left\{\dfrac{\underline p m}{\bar p M},1\right\}\right)^2 \l(A^*) \quad = \quad c\l_C(A) \]
where $c = {\underline p}^2 a^2\left(\min\left\{\dfrac{\underline p m}{\bar p M},1\right\}\right)^2 > 0$. \\ This completes the proof that $E$ is small.

That the chain is irreducible, follows easily, for any $\bx$, the set $\{\bx\}$ is a compact set and for a measurable set $A$ with $\l(A) > 0$ we may choose $C$ in the first part of the proof such that $\l(C\cap A) > 0$. Now,
\[ K^2(\bx,A) \quad \geq \quad c\l(C\cap A) > 0\]
Also aperiodicity follows trivially from the observation that any set with positive $\l$-measure can be accessed in at most 2 steps.
\end{proof}


\section{General TMCMC algorithm with single $\e$ and dependent $\bz$}
\label{sec:tmcmc_dependent_z}
Also, let 
Let $h_1(\bp)$, $h_2(\bq)$
be the specified joint distributions of $\bp$ and $\bq$ induced by the Gaussian distributions
of $\bw_1,\bw_2,\bw_3$, and
let $P(\bz\vert \bp,\bq)=\prod_{i=1}^kf_i(z_i\vert p_i,q_i)$ denote the conditional probability
of $\bz$, given $\bp,\bq$, where $f_i(\cdot\vert p_i,q_i)$ is the conditional
probability of $z_i$ given $p_i$ and $q_i$. 
Then the general TMCMC algorithm with singleton $\e$ and dependent $\bz$ is given as follows.

\begin{algo}\label{algo:GTM:multivar3} \topline General TMCMC algorithm based on single $\e$
and dependent $\bz$.
\botline \normalfont \ttfamily
\begin{itemize}
 \item Input: Initial value $\supr{\bm x}{0}$, and number of iterations $N$. 
\begin{enumerate}
 \item For $t=0,\ldots,N-1$
\begin{enumerate}
 \item Generate $\bw_1\sim N_k(\bmu_1,\bSigma_1)$, $\bw_2\sim N_k(\bmu_2,\bSigma_2)$, and $\bw_3\sim N_k(\mu_3,\bSigma_3)$. 
 \item For $i=1,\ldots,k$, set $p_i=\exp\left(w_{1i}\right)/\sum_{j=1}^3\exp\left(w_{ji}\right)$, 
 $q_i=\exp\left(w_{2i}\right)/\sum_{j=1}^3\exp\left(w_{ji}\right)$, and
 $1-p_i-q_i=\exp\left(w_{3i}\right)/\sum_{j=1}^3\exp\left(w_{ji}\right)$.
 \item Generate $\e \sim g(\cdot)$ and an index $i \sim \mathcal M(1;p_1,\ldots,p_{3^k-1})$ independently.
 \end{enumerate}
 \item \[ \bm x' = T_{\bz_i}(\supr{\bm x}{t}, \e) \quad \textrm{ and } 
 \quad \alpha(\supr{\bm x}{t}, \e) = \min\left(1, \dfrac{P(\bz^c_i\vert\bp,\bq)}{P(\bz_i\vert\bp,\bq)} 
 ~\dfrac{\pi(\bm x')}{\pi(\supr{\bm x}{t})} 
 ~\left|\frac{\partial (T_{\bz_i}(\supr{\bm x}{t}, \e),\e)}{\partial(\supr{\bm x}{t}, \e)}\right| \right)\]
\item Set \[ \supr{\bm x}{t+1}= \left\{\begin{array}{ccc}
 \bm x' & \textsf{ with probability } & \a(\supr{\bm x}{t},\e) \\
 \supr{\bm x}{t}& \textsf{ with probability } & 1 - \a(\supr{\bm x}{t},\e)
\end{array}\right.\]
\end{enumerate}
\item End for
\end{itemize}
\botline \rmfamily
\end{algo}

\section{Proof of detailed balance for TMCMC with dependent $\bz$}
\label{sec:detailed_balance_dependent_z}

Let $\by = T_{\bz}(\bx,\e) \in T_{\bz}(\bx,\Y)$, 
then $\bx = T_{\bz^c}(\by,\e)$. 
The kernel $K$ satisfies,
\begin{eqnarray*}
 \pi(\bx)K(\bx \to \by) & = & \pi(\bx)h_1(\bp)h_2(\bq)
 P(\bz\vert\bp,\bq)~g(\e)\min\left\{1,\dfrac{P(\bz^c\vert\bp,\bq) \pi(\by)}
 {P(\bz\vert\bp,\bq)\pi(\bx)}J_{\bz}(\bx,\e)\right\} \\
& = & h_1(\bp)h_2(\bq)g(\be)\min\left\{\pi(\bx)P(\bz\vert\bp,\bq),\pi(\by)P(\bz^c\vert\bp,\bq)
J_{\bz}(\bx,\e)\right\}
\end{eqnarray*}
and
\begin{eqnarray*}
 \pi(\by)K(\by \to \bx) & = & \pi(\by)h_1(\bp)h_2(\bq)P(\bz^c\vert\bp,\bq)g(\e)J_{\bz}(\bx,\e)
 \min\left\{1,\dfrac{P(\bz\vert\bp,\bq)
 \pi(\bx)}{P(\bz^c\vert\bp,\bq)\pi(\by)}J_{\bz^c}(\by,\e)~\right\} \\
& = & h_1(\bp)h_2(\bq)g(\e)\min\left\{\pi(\by)P(\bz^c\vert\bp,\bq)J_{\bz}(\bx,\e),\pi(\bx)P(\bz\vert\bp,\bq)
\right\}
\end{eqnarray*}

\section{Improved acceptance rates of additive TMCMC with singleton $\e$ compared 
to joint updating using RWMH}
\label{sec:accept_rate}

The joint RWMH algorithm generates $\be=(\e_1,\ldots,\e_k)'$ independently from $N(0,1)$,
and then uses the transformation of the form $x'_i=x_i+a_i\e_i$, where $a_i>0$ are
appropriate scaling constants. For large $k$, the so-called ``curse of dimensionality"
can force the acceptance rate to be close to zero. 
On the other hand, the additive-transformation based TMCMC also updates 
$(x_1,\ldots,x_k)$ simultaneously in a single block, but instead of using $k$ different $\e_i$,
it uses a single $\e$ for updating all the $x_i$ variables. In other words, for TMCMC
based on additive transformation $\be$ is of the form 
$\be=(\pm\e,\ldots,\pm\e)'$, where $\e\sim N(0,1)\mathbb I_{\{\e>0\}}$.
Thus, relative to RWMH, the dimension in the TMCMC case is effectively reduced to 1, 
avoiding the curse of dimensionality. Thus, it is expected that additive TMCMC will
have a much higher acceptance rate than RWMH.
In this section we formalize and compare the issues related to acceptance rates of additive TMCMC
and RWMH.

\subsection{Discussion on optimal scaling and optimal acceptance rate of additive TMCMC and RWMH}
\label{subsec:optimal_scaling}

A reasonable approach to compare the acceptance rates of additive TMCMC and RWMH is to 
develop the optimal scaling theory for additive TMCMC, obtain
the optimal acceptance rate, and then compare the latter with the optimal acceptance rates for RWMH,
which are already established in the MCMC literature. 
Indeed, optimal scaling and optimal acceptance rate of additive TMCMC and comparison with
those of RWMH is the subject of \ctn{Dey13b}, where it is shown that additive TMCMC has
a much higher optimal acceptance rate compared to RWMH. Before we summarize the results
of \ctn{Dey13b} we first provide a brief overview of optimal scaling and optimal acceptance rate of RWMH.

\subsubsection{Brief overview of optimal scaling and optimal acceptance rate for RWMH}  
\label{subsubsec:optimal_rwmh}

Roughly, the optimal random walk proposal variance, represented as an inverse function of the dimension $k$,
is the one that maximizes the speed of 
convergence to the stationary distribution of a relevant
diffusion process to which a `sped-up' version of RWMH weakly converges as the dimension $k$
increases to infinity. The optimal acceptance rate corresponds to the optimal proposal variance.
Under various assumptions on the form of the target distribution $\pi$, ranging from the
$iid$ assumption (\ctn{Roberts97}), through independent but non-identical set-up (\ctn{Bedard07}),
to a more general dependent structure (\ctn{Mattingly11}), the optimal acceptance rate
turns out to be 0.234. 

\subsubsection{Optimal scaling and optimal acceptance rate for additive TMCMC}
\label{subsubsec:optimal_tmcmc}

In \ctn{Dey13b} it has been proved
in the case of additive TMCMC, assuming $p_i=q_i=1/2$, 
that the optimal acceptance rate, as $k\rightarrow\infty$, is 0.439 
under the set-ups ($iid$, independent but non-identical, and dependent) for which the optimal 
acceptance rate for RWMH has been studied and established to be 0.234. 
Thus, the optimal acceptance rate for additive TMCMC is much higher than that of RWMH. 
The optimal scalings, that is, the optimal values of the scales $a_1,\ldots,a_k$ are also available
using the optimal scaling theory. 
As shown in \ctn{Dey13b}, all these results for additive TMCMC and RWMH 
remain true even in all the aforementioned set-ups if some of the co-ordinates of $\bx$ are
updated at random, conditioning on the remaining co-ordinates.   

\subsection{Comparison between the asymptotic forms of the acceptance rates of additive 
TMCMC and RWMH for strongly log-concave target densities}
\label{subsec:lower_upper_bound}
The results on optimal scaling and optimal acceptance rate discussed in Sections \ref{subsubsec:optimal_rwmh}
and \ref{subsubsec:optimal_tmcmc} are available only for special forms of the target distribution $\pi$. 
In this section we obtain the asymptotic forms of the acceptance rates associated with RWMH and additive TMCMC
assuming that the target density is strongly log-concave. 
In particular, under suitable conditions 
we show that as the dimension increases,
the acceptance rate of RWMH converges to zero at a much faster rate than that of additive TMCMC.

Assuming without loss of generality that the marginal variances of the target density $\pi$ are all unity 
(achieved after suitable scaling perhaps), for RWMH we consider the following proposed value $\bx'$ given
the current value $\bx$: $\bx'=\bx+\be$,
where $\be\sim N_k(\bzero,\bI_k)$. On the other hand, for additive TMCMC, we consider $\bx'=\bx+\e\bdelta$,
where $\epsilon \sim N(0,1)\mathbb{I}(\epsilon > 0)$ and the components $\delta_i$ 
of $\boldsymbol\delta$ are $iid$ taking values $\pm 1$ with probability $1/2$ each. 

To proceed we consider the following form of acceptance rate for our asymptotic framework.
Letting $R(\bx'|\bx)$ denote the acceptance probability of $\bx'$ given the current value $\bx$,
and letting $U\sim Uniform(0,1)$, the acceptance rate is given by
\begin{align}
AR &= \int R(\bx'|\bx)q(\bx'|\bx)\pi(\bx)d\bx d\bx'\notag
\\[1ex]
&=\int Pr\left(U<R(\bx'|\bx)\right)q(\bx'|\bx)\pi(\bx)d\bx d\bx'\notag
\\[1ex]
&=\int \left[\int Pr\left(U<R(\bx'|\bx)\right)q(\bx'|\bx)d\bx'\right]\pi(\bx)d\bx\notag
\\[1ex]
&=\int \left[\int_0^1 Pr\left(R(\bx'|\bx)>u\right)du\right]\pi(\bx)d\bx
\label{eq:ineq4}
\end{align}
In the above formula for acceptance rate note that $Pr\left(R(\bx'|\bx)>u\right)\rightarrow 1$ as $u\rightarrow 0$
and $Pr\left(R(\bx'|\bx)>u\right)\rightarrow 0$ as $u\rightarrow 1$.
Hence, given any $\eta_1>0,\eta_2>0$, we can choose $\psi_1,\psi_2\in (0,1)$ sufficiently small such that 
$\int_0^{\psi_1}Pr\left(R(\bx'|\bx)>u\right)du<\eta_1$ and $\int_{1-\psi_2}^1 Pr\left(R(\bx'|\bx)>u\right)du<\eta_2$.
Hence, re-writing (\ref{eq:ineq4})
as
\begin{eqnarray}
AR&=&\int \left[\int_0^{\psi_1} Pr\left(R(\bx'|\bx)>u\right)du\right]\pi(\bx)d\bx
+\int \left[\int_{\psi_1}^{1-\psi_2} Pr\left(R(\bx'|\bx)>u\right)du\right]\pi(\bx)d\bx\nonumber\\
&&\quad\quad\quad\quad+\int \left[\int_{1-\psi_2}^{1} Pr\left(R(\bx'|\bx)>u\right)du\right]\pi(\bx)d\bx,\nonumber
\end{eqnarray}
we find that the first and the third term on the right hand side are negligible for any algorithm. So,
for the purpose of comparing algorithms with respect to their acceptance rates, we consider only the middle term; 
in all that follow we denote 
\begin{equation}
AR = \int \left[\int_{\psi_1}^{1-\psi_2} Pr\left(R(\bx'|\bx)>u\right)du\right]\pi(\bx)d\bx.
\label{eq:ineq5}
\end{equation}

For our purpose, we consider a target density $\pi(\mathbf{x})$ of $k$ variables that is strongly log-concave, 
that is,
\begin{equation}\label{eqn:stronglyconvex}
 -M_k \mathbf{I}_k \leq \nabla^2 \log \pi(\mathbf{x}) \leq -m_k \mathbf{I}_k,
\end{equation}
where we assume that $M_k=c_k+m_k$, with $m_k,~c_k>0$ for every $k$. We further assume that 
$m_k\rightarrow\infty$, and the sequence $\{c_k\}$ is such that $c_k/m_k \rightarrow 0$ as $k\rightarrow\infty$. 
Then clearly, $M_k\asymp m_k$, meaning $M_k/m_k\rightarrow 1$ as $k\rightarrow\infty$. In fact, we assume that
$M_k/m_k$ approaches 1 at a sufficiently fast rate, so that
$k\left\vert\frac{M_k}{m_k}-1\right\vert\rightarrow 0$. For our purpose we assume that 
$c_k=O(k^s); s\geq 1$ and 
$m_k=O(k^t);t>s+1\geq 2$,
so that $M_k=O(k^t)$. It is easy to verify that these choices satisfy the above conditions.

It is important to note that our assumption $m_k,M_k\rightarrow\infty$ need not hold for all
strongly log-concave distributions. For instance,
when $\pi$ is the $iid$ product of standard normals, that is, when 
$\bx\sim N_k\left(\bzero,\bI_k\right)$ under $\pi$, $\nabla^2 \log \pi(\mathbf{x})=\bI_k$.
In this case, $m_k=M_k=1$ for every $k\geq 1$. In general, even if $m_k$ and $M_k$ remains finite
as $k$ grows to infinity, our proofs remain valid provided that $M_k\asymp m_k$ and 
$k\left\vert\frac{M_k}{m_k}-1\right\vert\rightarrow 0$. The case of $\pi$ being an $iid$ product
of standard normals clearly satisfies the above conditions.

\subsubsection{Asymptotic form of the acceptance rate for RWMH}
\label{subsubsec:asymp_rwmh}

Let $\bx^*$ denote the mode of the target density $\pi(\cdot)$.
Then for every $r \in (0,1),$
\begin{eqnarray*}
& & Pr~(R(\mathbf{x}'| \mathbf{x}) < r)  = Pr~(\pi(\mathbf{x}')/\pi(\mathbf{x}) < r) = Pr~(\log \pi(\mathbf{x}') - \log \pi(\mathbf{x})  < \log r )\\
& = &  Pr~(\left[\log \pi(\mathbf{x}') - \log \pi(\mathbf{x^*})\right]-\left[\log \pi(\mathbf{x}') - \log \pi(\mathbf{x^*})\right]  < \log r )\\
& = & Pr~\left(\left[\nabla \log \pi(\mathbf{x^*})^T (\mathbf{x}' - \mathbf{x^*}) + (1/2)(\mathbf{x}'-\mathbf{x^*})^T \nabla^2 \log \pi(\boldsymbol{\xi_1(\mathbf x',\mathbf x^*)})(\mathbf{x}'-\mathbf{x^*})\right]\right.\\
&&\quad\quad\left.-\left[\nabla \log \pi(\mathbf{x^*})^T (\mathbf{x} - \mathbf{x^*}) + (1/2)(\mathbf{x}-\mathbf{x^*})^T \nabla^2 \log \pi(\boldsymbol{\xi_2(\mathbf x,\mathbf x^*)})(\mathbf{x}-\mathbf{x^*})\right]
< \log r\right), \\
&&\quad\textrm{~for some } \boldsymbol\xi_1(\mathbf x',\mathbf x^*), \boldsymbol\xi_2(\mathbf x,\mathbf x^*)\ \ \mbox{depending upon}\ \ (\mathbf x',\mathbf x^*) \ \ \mbox{and} (\mathbf x,\mathbf x^*) \ \ \mbox{respectively};\\
& = & Pr~\left(\left[(1/2)(\mathbf{x}'-\mathbf{x^*})^T \nabla^2 \log \pi(\boldsymbol{\xi_1(\mathbf x',\mathbf x^*)})(\mathbf{x}'-\mathbf{x^*})\right]\right.\\
&&\quad\quad\left. -\left[(1/2)(\mathbf{x}-\mathbf{x^*})^T \nabla^2 \log \pi(\boldsymbol{\xi_2(\mathbf x',\mathbf x^*)})(\mathbf{x}-\mathbf{x^*})\right]<\log r\right)\\
&&\quad\quad\quad\quad\quad\mbox{since}\ \ \nabla \log \pi(\mathbf{x^*})=\bzero.
\end{eqnarray*}
Thus from the assumption in \eqref{eqn:stronglyconvex}, and noting that
$(\mathbf{x}'-\mathbf{x^*})^T(\mathbf{x}'-\mathbf{x^*})=(\mathbf{x}-\mathbf{x^*})^T(\mathbf{x}-\mathbf{x^*})
+2(\mathbf{x}-\mathbf{x^*})^T\boldsymbol\epsilon+\boldsymbol\epsilon^T\boldsymbol\epsilon$
it follows that 
\begin{equation}\label{eqn:bound1}
 \begin{split}
  & Pr~\left(\frac{(M_k-m_k)}{2}(\mathbf{x}-\mathbf{x^*})^T (\mathbf{x}-\mathbf{x^*}) 
  -m_k(\mathbf{x}-\mathbf{x^*})^T\boldsymbol\epsilon - \frac{m_k}{2}\boldsymbol\epsilon^T\boldsymbol\epsilon
  < \log r\right) \\
 & \leq Pr~(R(\mathbf{x}'| \mathbf{x}) < r) \\
 & \leq 
   Pr~\left(-\frac{(M_k-m_k)}{2}(\mathbf{x}-\mathbf{x^*})^T (\mathbf{x}-\mathbf{x^*}) 
  -M_k(\mathbf{x}-\mathbf{x^*})^T\boldsymbol\epsilon - \frac{M_k}{2}\boldsymbol\epsilon^T\boldsymbol\epsilon
  < \log r\right);
 \end{split}
\end{equation}
so that
\begin{equation}\label{eqn:new_bound2}
 \begin{split}
  & Pr~\left(\frac{(M_k-m_k)}{2}(\mathbf{x}-\mathbf{x^*})^T (\mathbf{x}-\mathbf{x^*}) 
  -m_k(\mathbf{x}-\mathbf{x^*})^T\boldsymbol\epsilon - \frac{m_k}{2}\boldsymbol\epsilon^T\boldsymbol\epsilon
  > \log r\right) \\
 & \geq Pr~(R(\mathbf{x}'| \mathbf{x}) > r) \\
 & \geq 
   Pr~\left(-\frac{(M_k-m_k)}{2}(\mathbf{x}-\mathbf{x^*})^T (\mathbf{x}-\mathbf{x^*}) 
  -M_k(\mathbf{x}-\mathbf{x^*})^T\boldsymbol\epsilon - \frac{M_k}{2}\boldsymbol\epsilon^T\boldsymbol\epsilon
  > \log r\right),
 \end{split}
\end{equation}
Hence, using (\ref{eq:ineq5}) it can be seen that the acceptance rate is bounded above and below as follows
\begin{align}
& \int \left[\int_{\psi_1}^{1-\psi_2} \left\{ \int_{A^k_{2,\be,u}}\frac{1}{(2\pi)^{k/2}}\exp\left\{-\frac{1}{2}\be^T\be\right\}d\be \right\}du\right]\pi(\bx)d\bx\notag
\\[1ex]
&\geq AR^{(RWMH)}\label{eq:ar_rwmh} 
\\[1ex]
&\geq \int \left[\int_{\psi_1}^{1-\psi_2} \left\{ \int_{A^k_{1,\be,u}}\frac{1}{(2\pi)^{k/2}}\exp\left\{-\frac{1}{2}\be^T\be\right\}d\be \right\}du\right]\pi(\bx)d\bx,
\notag
\end{align}
where
$$A^k_{1,\be,u}=\left\{\mathbf{x}: -\frac{(M_k-m_k)}{2}(\mathbf{x}-\mathbf{x^*})^T (\mathbf{x}-\mathbf{x^*})
-M_k(\mathbf{x}-\mathbf{x^*})^T\boldsymbol\epsilon - \frac{M_k}{2}\boldsymbol\epsilon^T\boldsymbol\epsilon
> \log u\right\}$$ and
$$A^k_{2,\be,u}=\left\{\mathbf{x}: \frac{(M_k-m_k)}{2}(\mathbf{x}-\mathbf{x^*})^T (\mathbf{x}-\mathbf{x^*})
-m_k(\mathbf{x}-\mathbf{x^*})^T\boldsymbol\epsilon - \frac{m_k}{2}\boldsymbol\epsilon^T\boldsymbol\epsilon
> \log u\right\}.$$
Now, note that for some $\bxi(\bx,\bx^*)$ depending upon $\bx$ and $\bx^*$,
\begin{eqnarray}
\pi(\bx)
 &=&\exp\left\{\log\pi(\mathbf{x})\right\}d\mathbf{x}\nonumber\\
 &=& \exp\left\{\log\pi(\mathbf{x^*})
 +\frac{1}{2}(\mathbf{x}-\mathbf{x^*})^T\nabla^2\log\pi(\boldsymbol\xi(\mathbf x,\mathbf x^*))
 (\mathbf{x}-\mathbf{x^*})\right\}, \nonumber\\
\end{eqnarray}
so that the inequalities related to strong convexity, given by (\ref{eqn:stronglyconvex}) yield
\begin{equation*}
\frac{(2\pi)^{k/2}}{m^k_k}\pi(\mathbf x^*)
\frac{m^k_k}{(2\pi)^{k/2}} \exp\left\{-\frac{m_k}{2}(\mathbf{x}-\mathbf{x^*})^T(\mathbf{x}-\mathbf{x^*})\right\}
\geq \pi(\bx)
\end{equation*}
\begin{equation}
\geq \frac{(2\pi)^{k/2}}{M^k_k}\pi(\mathbf x^*)
\frac{M^k_k}{(2\pi)^{k/2}} \exp\left\{-\frac{M_k}{2}(\mathbf{x}-\mathbf{x^*})^T(\mathbf{x}-\mathbf{x^*})\right\}
\label{eq:eq2}
\end{equation}

Using the lower bound of $\pi(\bx)$ given by (\ref{eq:eq2}) and Fubini's theorem, the lower bound
of the acceptance rate given by (\ref{eq:ar_rwmh}) can be further bounded below as
\begin{align}
AR^{(RWMH)} &\geq \int \int_{\psi_1}^{1-\psi_2}  \int_{A^k_{1,\be,u}}\frac{1}{(2\pi)^{k/2}}\exp\left\{-\frac{1}{2}\be^T\be\right\}\pi(\bx)~d\bx~du~d\be \notag
\\[1ex]
&\geq \frac{(2\pi)^{k/2}}{M^k_k}\pi(\mathbf x^*)\int \int_{\psi_1}^{1-\psi_2}  \int_{A^k_{1,\be,u}}\frac{1}{(2\pi)^{k/2}}\exp\left\{-\frac{1}{2}\be^T\be\right\}\notag
\\[1ex]
&\quad\quad\quad\quad\times\frac{M^k_k}{(2\pi)^{k/2}} \exp\left\{-\frac{M_k}{2}(\mathbf{x}-\mathbf{x^*})^T(\mathbf{x}-\mathbf{x^*})\right\}~d\bx~du~d\be\notag
\\[1ex]
&\geq \frac{(2\pi)^{k/2}}{M^k_k}\pi(\mathbf x^*)\inf_{u\in (\psi_1,1-\psi_2)}Pr~(A^k_{1,\be,u}),
\label{eq:lowerbound1}
\end{align}
where $Pr~(A^k_{1,\be,u})$ must be calculated with respect to $\be\sim N_k(\bzero,\bI_k)$, and
independently, $\bx-\bx^*\sim N_k(\bzero,M^{-1}_k\bI_k)$. 

Similarly, using the upper bound of $\pi(\bx)$ given by (\ref{eq:eq2}) the upper bound
of the acceptance rate given by (\ref{eq:ar_rwmh}) can be further bounded above as
\begin{align}
AR^{(RWMH)} &\leq \int \int_{\psi_1}^{1-\psi_2}  \int_{A^k_{2,\be,u}}\frac{1}{(2\pi)^{k/2}}\exp\left\{-\frac{1}{2}\be^T\be\right\}\pi(\bx)~d\bx~du~d\be \notag
\\[1ex]
&= \frac{(2\pi)^{k/2}}{m^k_k}\pi(\mathbf x^*)\int_{\psi_1}^{1-\psi_2} Pr~(A^k_{2,\be,u})du\notag
\\[1ex]
&\leq \frac{(2\pi)^{k/2}}{m^k_k}\pi(\mathbf x^*)\sup_{u\in (\psi_1,1-\psi_2)}Pr~(A^k_{2,\be,u})
\label{eq:upperbound1}
\end{align}
The probability 
$Pr~(A^k_{2,\be,u_k})$ must be calculated with respect to $\be\sim N_k(\bzero,\bI_k)$, and
independently, $\bx-\bx^*\sim N_k(\bzero,m^{-1}_k\bI_k)$. 
Thus, we have
\begin{equation}
\frac{(2\pi)^{k/2}}{M^k_k}\pi(\mathbf x^*)\inf_{u\in (\psi_1,1-\psi_2)}Pr~(A^k_{1,\be,u})
\leq AR^{(RWMH)}
\leq \frac{(2\pi)^{k/2}}{m^k_k}\pi(\mathbf x^*)\sup_{u\in (\psi_1,1-\psi_2)}Pr~(A^k_{2,\be,u}).
\label{eq:lower_upper1}
\end{equation}

We first focus on the lower bound in (\ref{eq:lower_upper1}).
As $k\rightarrow\infty$,
\\[2mm]
$-\frac{(M_k-m_k)}{2}(\mathbf{x}-\mathbf{x^*})^T (\mathbf{x}-\mathbf{x^*})
-M_k(\mathbf{x}-\mathbf{x^*})^T\boldsymbol\epsilon - \frac{M_k}{2}\boldsymbol\epsilon^T\boldsymbol\epsilon$
\begin{eqnarray}
&\sim &
AN\left(-\frac{k}{2}\left[\left(\frac{M_k-m_k}{M_k}\right)+M_k\right],
\frac{k}{2}\left[\left(\frac{M_k-m_k}{M_k}\right)^2+2M_k+M^2_k\right]\right),
\label{eq:an1}
\end{eqnarray}
where $AN(\mu,\sigma^2)$ denotes 
asymptotic normal with mean $\mu$ and variance $\sigma^2$. From (\ref{eq:an1}) it follows that
\begin{eqnarray}
\inf_{u\in (\psi_1,1-\psi_2)}Pr~(A^k_{1,\be,u}) &\asymp &  
1-\sup_{u\in (\psi_1,1-\psi_2)}
\Phi\left(\frac{\log u+\frac{k}{2}\left[\left(\frac{M_k-m_k}{M_k}\right)+M_k\right]}
{\sqrt{\frac{k}{2}\left[\left(\frac{M_k-m_k}{M_k}\right)^2+2M_k+M^2_k\right]}}\right)\nonumber\\
&=& 1- \Phi\left(\frac{\log (1-\psi_2)+\frac{k}{2}\left[\left(\frac{M_k-m_k}{M_k}\right)+M_k\right]}
{\sqrt{\frac{k}{2}\left[\left(\frac{M_k-m_k}{M_k}\right)^2+2M_k+M^2_k\right]}}\right).
\label{eq:infimum1}
\end{eqnarray}
Combining (\ref{eq:lowerbound1}) and (\ref{eq:infimum1}) we obtain
\begin{equation}
AR^{(RWMH)}\geq
\frac{(2\pi)^{k/2}}{M^k_k}\pi(\mathbf x^*)
\left\{1-\Phi\left(\frac{\log (1-\psi_2)+\frac{k}{2}\left[\left(\frac{M_k-m_k}{M_k}\right)+M_k\right]}
{\sqrt{\frac{k}{2}\left[\left(\frac{M_k-m_k}{M_k}\right)^2+2M_k+M^2_k\right]}}\right)
\right\}.
\label{eq:AR_RWMH_lowerbound1}
\end{equation}

Now focusing our attention on the upper bound of $AR^{(RWMH)}$ we similarly obtain
\begin{eqnarray}
AR^{(RWMH)}
&\leq &\frac{(2\pi)^{k/2}}{m^k_k}\pi(\mathbf x^*)
\left\{1-
\Phi\left(\frac{\log \psi_1-\frac{k}{2}\left[\left(\frac{M_k-m_k}{m_k}\right)-m_k\right]}
{\sqrt{\frac{k}{2}\left[\left(\frac{M_k-m_k}{m_k}\right)^2+2m_k+m^2_k\right]}}\right)\right\}.\nonumber\\
\label{eq:AR_RWMH_upperbound1}
\end{eqnarray}
In other words,
\begin{align}\label{eq:AR_RWMH_bounds}
\frac{(2\pi)^{k/2}}{M^k_k}\pi(\mathbf x^*)
&\left\{1-\Phi\left(\frac{\log (1-\psi_2)+\frac{k}{2}\left[\left(\frac{M_k-m_k}{M_k}\right)+M_k\right]}
{\sqrt{\frac{k}{2}\left[\left(\frac{M_k-m_k}{M_k}\right)^2+2M_k+M^2_k\right]}}\right)
\right\}
\leq  AR^{(RWMH)}\notag\\
&\leq \frac{(2\pi)^{k/2}}{m^k_k}\pi(\mathbf x^*)
\left\{1-\Phi\left(\frac{\log \psi_1-\frac{k}{2}\left[\left(\frac{M_k-m_k}{m_k}\right)-m_k\right]}
{\sqrt{\frac{k}{2}\left[\left(\frac{M_k-m_k}{m_k}\right)^2+2m_k+m^2_k\right]}}\right)\right\}.\nonumber\\
\end{align}
Since $m_k\asymp M_k$, it is easy to see that 
\begin{align}\label{eq:AR_RWMH_bounds2}
&\frac{\log (1-\psi_2)+\frac{k}{2}\left[\left(\frac{M_k-m_k}{M_k}\right)+M_k\right]}
{\sqrt{\frac{k}{2}\left[\left(\frac{M_k-m_k}{M_k}\right)^2+2M_k+M^2_k\right]}}
\asymp  \sqrt{\frac{k}{2}},\quad\mbox{and}\notag\\
&\frac{\log \psi_1-\frac{k}{2}\left[\left(\frac{M_k-m_k}{m_k}\right)-m_k\right]}
{\sqrt{\frac{k}{2}\left[\left(\frac{M_k-m_k}{m_k}\right)^2+2m_k+m^2_k\right]}}
\asymp  \sqrt{\frac{k}{2}}.\notag
\end{align}
Hence, it follows that
\begin{equation}
AR^{(RWMH)}\asymp
\frac{(2\pi)^{k/2}}{M^k_k}\left\{1-\Phi\left(\sqrt{\frac{k}{2}}\right)\right\}.
\label{eq:RWMH_asymp}
\end{equation}

\subsubsection{Asymptotic bounds of the acceptance rate for additive TMCMC}
\label{subsubsec:asymp_tmcmc}
Next let us obtain lower and upper bounds of $AR^{(TMCMC)}$ associated with TMCMC with additive transformation. 
Recall that in this case, $\mathbf{x}' = \mathbf{x} + \epsilon \boldsymbol\delta$ 
where $\epsilon \sim N(0,1)\mathbb{I}(\epsilon > 0)$ and the components $\delta_i$ 
of $\boldsymbol\delta$ are $iid$ taking values $\pm1$ with probability $1/2$ each. 
In this set up (\ref{eqn:new_bound2}) becomes 
\begin{equation}\label{eqn:new_bound3}
 \begin{split}
  & Pr~\left(\frac{(M_k-m_k)}{2}(\mathbf{x}-\mathbf{x^*})^T (\mathbf{x}-\mathbf{x^*}) 
  -m_k\e(\mathbf{x}-\mathbf{x^*})^T\bdelta - \frac{m_k}{2}k\e^2
  > \log r\right) \\
 & \geq Pr~(R(\mathbf{x}'| \mathbf{x}) > r) \\
 & \geq 
   Pr~\left(-\frac{(M_k-m_k)}{2}(\mathbf{x}-\mathbf{x^*})^T (\mathbf{x}-\mathbf{x^*}) 
  -M_k\e(\mathbf{x}-\mathbf{x^*})^T\bdelta - \frac{M_k}{2}k\e^2
  > \log r\right),
 \end{split}
\end{equation}
Now notice that, under the lower bound of $\pi(\bx)$ provided in (\ref{eq:eq2}), as $k\rightarrow\infty$, 
$$\frac{M_k(\mathbf{x}-\mathbf{x^*})^T (\mathbf{x}-\mathbf{x^*})}{k}
\stackrel{\a. s.}{\longrightarrow} 1,$$ and
$$\frac{\sqrt{M_k}(\mathbf{x}-\mathbf{x^*})^T\bdelta}{k} \stackrel{\a. s.}{\longrightarrow}0.$$
Similarly, under the upper bound of $\pi(\bx)$ in (\ref{eq:eq2}), the above hold with $M_k$ replaced with $m_k$.
From these it follow that the asymptotic forms of the 
lower and the upper bounds of (\ref{eqn:new_bound3}) are given by
\[
Pr~\left(-\frac{(M_k-m_k)}{2}(\mathbf{x}-\mathbf{x^*})^T (\mathbf{x}-\mathbf{x^*}) 
  -M_k\e(\mathbf{x}-\mathbf{x^*})^T\bdelta - \frac{M_k}{2}k\e^2
  > \log r\right)
\]
\[ \asymp
\frac{(2\pi)^{k/2}}{M^k_k}\pi(\bx^*) 
\left\{2\Phi\left(\sqrt{-\frac{2}{kM_k}\log r-\left(\frac{M_k-m_k}{M^2_k}\right)}\right)-1\right\}
\]
and
\[
Pr~\left(\frac{(M_k-m_k)}{2}(\mathbf{x}-\mathbf{x^*})^T (\mathbf{x}-\mathbf{x^*}) 
  -m_k\e(\mathbf{x}-\mathbf{x^*})^T\bdelta - \frac{m_k}{2}k\e^2
  > \log r\right)
\]
\[ \asymp
\frac{(2\pi)^{k/2}}{m^k_k}\pi(\bx^*)
\left\{2\Phi\left(\sqrt{-\frac{2}{km_k}\log r+\left(\frac{M_k-m_k}{m^2_k}\right)}\right)-1\right\}. 
\]
Using the above results, it follows as in the case of $AR^{(RWMH)}$ that
\begin{align} 
&\frac{(2\pi)^{k/2}}{M^k_k}\pi(\bx^*)
\left\{2\inf_{u\in (\psi_1,1-\psi_2)}\Phi\left(\sqrt{-\frac{2}{kM_k}\log u-\left(\frac{M_k-m_k}{M^2_k}\right)}\right)-1\right\}\notag\\
&\leq AR^{(TMCMC)}
\leq \frac{(2\pi)^{k/2}}{M^k_k}\pi(\bx^*) 
\left\{2\sup_{u\in (\psi_1,1-\psi_2)}\Phi\left(\sqrt{-\frac{2}{km_k}\log u-\left(\frac{m_k-M_k}{m^2_k}\right)}\right)-1\right\}.\notag
\end{align}
Substituting the infimum and supremum over $(\psi_1,1-\psi_2)$ we obtain 
\begin{align} 
&\frac{(2\pi)^{k/2}}{M^k_k}\pi(\bx^*)
\left\{2\Phi\left(\sqrt{-\frac{2}{kM_k}\log (1-\psi_2)-\left(\frac{M_k-m_k}{M^2_k}\right)}\right)-1\right\}\notag\\
&\leq AR^{(TMCMC)}
\leq \frac{(2\pi)^{k/2}}{m^k_k}\pi(\bx^*)  
\left\{2\Phi\left(\sqrt{-\frac{2}{km_k}\log \psi_1-\left(\frac{m_k-M_k}{m^2_k}\right)}\right)-1\right\}.\notag
\end{align}
Since $k\left\vert\frac{M_k}{m_k}-1\right\vert\rightarrow 0$ and $m_k\asymp M_k$, it follows that
\begin{align} 
& -\frac{2}{kM_k}\log (1-\psi_2)-\left(\frac{M_k-m_k}{M^2_k}\right)
\asymp 
-\frac{2}{kM_k}\log (1-\psi_2)\quad\mbox{and}\notag\\
& -\frac{2}{km_k}\log \psi_1-\left(\frac{m_k-M_k}{m^2_k}\right)
\asymp 
-\frac{2}{km_k}\log \psi_1\asymp -\frac{2}{kM_k}\log \psi_1.\notag
\end{align}
Hence,  
\begin{eqnarray}
\frac{(2\pi)^{k/2}}{M^k_k}\pi(\bx^*)
\left\{2\Phi\left(\sqrt{-\frac{2}{kM_k}\log (1-\psi_2)}\right)-1\right\}
&\leq & AR^{(TMCMC)}\nonumber\\
&\leq & \frac{(2\pi)^{k/2}}{M^k_k}\pi(\bx^*)  
\left\{2\Phi\left(\sqrt{-\frac{2}{kM_k}\log \psi_1}\right)-1\right\}.\nonumber\\
\label{eq:AR_TMCMC_bounds}
\end{eqnarray}

For comparing (\ref{eq:AR_TMCMC_bounds}) with (\ref{eq:RWMH_asymp}) 
where $M_k=O\left(k^t\right);t>2$, it can be easily verified using L'Hospital's rule that
for any $\zeta_1>0$, $\zeta_2>0$,
\begin{equation}
\frac{2\Phi\left(\frac{\zeta_1}{\sqrt{kM_k}}\right)-1}{1-\Phi\left(\zeta_2\sqrt{k}\right)}
\rightarrow \infty.
\label{eq:comparison1}
\end{equation}
The above result will continue to hold if instead of $M_k=O\left(k^t\right);t>2$, 
$M_k\rightarrow a$, where $a>0$ is some constant.
Hence, $AR^{(TMCMC)}$
converges to zero at a much slower rate compared to $AR^{(RWMH)}$.

\section{Comparison of TMCMC with HMC}
\label{sec:hmc}

Motivated by Hamiltonian dynamics, \ctn{Duane87} introduced HMC, an MCMC algorithm with 
 deterministic proposals based on approximations of the Hamiltonian equations. 
 We will show that this algorithm is a special case of TMCMC, but first we provide a brief
 overview of HMC. More details can be found in \ctn{Liu01}, \ctn{Cheung09} and the references therein.

\subsection{Overview of HMC}
\label{subsec:hmc_overview}

If $\pi(\bx)$ is the target distribution, a fictitious dynamical system may be considered, where
$\bx(t)\in\mathbb R^d$ can be thought of as the $d$-dimensional position vector of a body of particles
at time $t$. If $\bv(t)=\dot{\bx}(t)=\frac{d{\bx}}{dt}$ is the speed vector of the
particles, $\dot{\bv}(t)=\frac{d{\bv}}{dt}$ is its acceleration vector, and 
$\Vec{F}$ is the force exerted on the particle; then, by Newton's law of motion 
$\Vec{F}=\boldm\dot{\bv}(t)=(m_1\dot{v_1},\ldots,m_d\dot{v_d})(t)$, where $\boldm\in\mathbb R^d$
is a mass vector. The momentum vector, $\bp=\boldm\bv$, often used in classical mechanics,
can be thought of as a vector of auxiliary variables brought in to facilitate simulation from $\pi(\bx)$.
The kinetic energy of the system is defined as $W(\bp)=\bp'\bM^{-1}\bp$, $\bM$ being the mass matrix.
Usually, $\bM$ is taken as $\bM=diag\{m_1,\ldots,m_d\}$.

The target density $\pi(\bx)$ is linked to the dynamical system via the potential energy field
of the system, defined as $U(\bx)=-\log\pi(\bx)$. The total energy (Hamiltonian function), is given by
$H(\bx,\bp)=U(\bx)+W(\bp)$.
A joint distribution over the phase-space $(\bx,\bp)$ is then considered, given by
\begin{equation}
f(\bx,\bp)\propto\exp\left\{-H(\bx,\bp)\right\}=\pi(\bx)\exp\left(-\bp'\bM^{-1}\bp/2\right)
\label{eq:phase_space_dist}
\end{equation}
Since the marginal density of $f(\bx,\bp)$ is $\pi(\bx)$, it now remains to provide a joint proposal
mechanism for simulating $(\bx,\bp)$ jointly; ignoring $\bp$ yields $\bx$ marginally from $\pi(\cdot)$. 

For the joint proposal mechanism, HMC makes use of Newton's law of motion, derived from
the law of conservation of energy, and often written in the form of Hamiltonian equations, given by 
\begin{eqnarray}
\dot{\bx}(t)&=&\frac{\partial H(\bx,\bp)}{\partial\bp}=\bM^{-1}\bp,\nonumber\\
\dot{\bp}(t)&=&-\frac{\partial H(\bx,\bp)}{\partial\bx}=-\nabla U(\bx),\nonumber
\end{eqnarray}
where $\nabla U(\bx)=\frac{\partial U(\bx)}{\partial\bx}$.
The Hamiltonian equations can be approximated by the commonly used leap-frog algorithm (\ctn{Hockney70}), 
given by,
\begin{align}
\bx(t+\d t)&=\bx(t)+\d t\bM^{-1}\left\{\bp(t)-\frac{\d t}{2}\nabla U\left(\bx(t)\right)\right\}\label{eq:frog1}
\\[1ex]
\bp(t+\d t)&=\bp(t)-\frac{\d t}{2}\left\{\nabla U\left(\bx(t)\right)+\nabla U\left(\bx(t+\d t)\right)\right\}\label{eq:frog2}
\end{align}
Given choices of $\bM$, $\d t$, and $L$, the HMC is then given by the following algorthm:
\begin{algo}\label{algo:hmc} \topline HMC \botline \normalfont \ttfamily
\begin{itemize}
 \item Initialise $\bx$ and draw $\bp\sim N(\bzero, \bM)$.
 \item Assuming the current state to be $(\bx,\bp)$, do the following:
\begin{enumerate}
 \item Generate $\bp'\sim N\left(\bzero,\bM\right)$;
 \item Letting $(\bx(0),\bp(0))=(\bx,\bp')$, run the leap-frog algorithm for 
 $L$ time steps, to yield $(\bx'',\bp'')=\left(\bx(t+L\d t),\bp(t+L\d t)\right)$;
\item Accept $(\bx'',\bp'')$ with probability 
\begin{equation}
\min\left\{1,\exp\left\{-H(\bx'',\bp'')+H(\bx,\bp')\right\}\right\},
\label{eq:hmc_accept}
\end{equation}
and accept $(\bx,\bp')$ with the remaining probability.
\end{enumerate}
\end{itemize}
\botline \rmfamily
\end{algo}
In the above algorithm, it is not required to store simulations of $\bp$.
Next we show that HMC is a special case of TMCMC.

\subsection{HMC is a special case of TMCMC}
\label{subsec:hmcmc_special_tmcmc}

To see that HMC is a special case of TMCMC, note that the leap-frog step of the HMC algorithm (Algorithm \ref{algo:hmc})
is actually a deterministic transformation of the form $g^L: (\bx(0),\bp(0))\to(\bx(L),\bp(L))$ (see \ctn{Liu01}). 
This transformation satisfies the following: if $(\bx',\bp')=g^L(\bx,\bp)$, then $(\bx,-\bp)=g^L(\bx',-\bp')$.

The Jacobian
of this transformation is 1 because of the volume preservation property, which says that if
$V(0)$ is a subset of the phase space, and if $V(t)=\left\{(\bx(t),\bp(t)): (\bx(0),\bp(0))\in V(0)\right\}$, then the volume
$|V(t)|={\int\int}_{V(t)}d\bx d\bp={\int\int}_{V(0)}d\bx d\bp=|V(0)|$. As a result, the Jacobian does not feature in the
HMC acceptance probability (\ref{eq:hmc_accept}).

For any dimension, there is only one move type defined for HMC,
which is the forward transformation $g^L$. Hence, this move type has probability one of selection, and all
other move types which we defined in general terms in connection with TMCMC, have zero probability of selection. 
As a result, 
the corresponding TMCMC acceptance ratio needs 
slight modification---it must be made free of the move-type probabilities, 
which is exactly the case in (\ref{eq:hmc_accept}).

The momentum vector
$\bp$ can be likened to $\be$ of TMCMC, but note that $\bp$ must always be of the same dimensionality as $\bx$;
this is of course, permitted by TMCMC as a special case. 

\subsection{Comparison of acceptance rate for $L=1$ with RWMH and TMCMC}
\label{subsec:hmc_vs_tmcmc}

For $L=1$, the proposal corresponding to HMC is given by (see \ctn{Cheung09})
\begin{equation}
q(\bx'\mid\bx(t))=N\left(\bx':\bmu(t),\bSigma(t)\right),
\label{eq:hmc_proposal}
\end{equation}
where (\ref{eq:hmc_proposal}) is a normal distribution with mean and variance
given, respectively, by the following:
\begin{align}
\bmu(t)&=\bx(t)+\frac{1}{2}\bM^{-1}\d t\nabla\log\left(\pi(\bx(t))\right)\label{eq:hmc_mean}
\\[1ex]
\bSigma(t)&=\d t\bM^{-1}\label{eq:hmc_var}
\end{align}
Assuming diagonal $\bM$ with $m_i$ being the $i$-th diagonal element,
the proposal can be re-written in the following more convenient manner:
for $i=1,\ldots,k$,
\begin{equation}
x'_i=x_i(t)+\e_i,
\label{eq:hmc_proposal2}
\end{equation}
where $s_i(t)$ denotes the $i$-th component of $\nabla\log\left(\pi(\bx(t))\right)$,
and $\e_i\sim N\left(\frac{1}{2}\frac{\d t s_i(t)}{m_i},\frac{\d t}{m_i}\right)$.
Assuming, as is usual, that $m_i=1$ for each $i$, it follows that
\begin{equation}
\frac{\parallel\bx'-\bx\parallel^2}{\d t^2} =\sum_{i=1}^k\left(\frac{\e_i}{\d t}\right)^2=\sum_{i=1}^k\e'^2_i
\sim\chi^2_k(\lambda),
\label{eq:hmc_rate1}
\end{equation}
where $\chi^2_k(\lambda)$ is a non-central $\chi^2$ distribution with $k$ degrees of freedom
and non-centrality parameter 
$\lambda =\frac{\d t^2}{4}\sum_{i=1}^ks^2_i(t)$.
Since, as either $k\rightarrow\infty$ or $\lambda\rightarrow\infty$,
\begin{equation}
\frac{\chi^2_k(\lambda)-(k+\lambda)}{\sqrt{2(k+2\lambda)}}\stackrel{\mathcal L}{\rightarrow} N(0,1),
\label{eq:clt}
\end{equation}
assuming the same strong log-concavity conditions on the target density $\pi$ 
as provided in Section \ref{subsec:lower_upper_bound} it follows as in (\ref{eq:AR_RWMH_bounds}) that,
\begin{align}\label{eq:AR_HMC_bounds}
\frac{(2\pi)^{k/2}}{M^k_k}\pi(\mathbf x^*)
&\left\{1-\Phi\left(\frac{\log (1-\psi_2)+\frac{k}{2}\left[\left(\frac{M_k-m_k}{M_k}\right)+M_k\d t^2\left(1+\frac{\lambda}{k}\right)\right]}
{\sqrt{\frac{k}{2}\left[\left(\frac{M_k-m_k}{M_k}\right)^2+2M_k\d t^2\left(1+\frac{\lambda}{k}\right)
+M^2_k\d t^4\left(1+\frac{2\lambda}{k}\right)\right]}}\right)
\right\}
\leq  AR^{(HMC)}\notag\\
&\leq \frac{(2\pi)^{k/2}}{m^k_k}\pi(\mathbf x^*)
\left\{1-\Phi\left(\frac{\log \psi_1-\frac{k}{2}\left[\left(\frac{M_k-m_k}{m_k}\right)-m_k\d t^2\left(1+\frac{\lambda}{k}\right)\right]}
{\sqrt{\frac{k}{2}\left[\left(\frac{M_k-m_k}{m_k}\right)^2+2m_k\d t^2\left(1+\frac{\lambda}{k}\right)
+m^2_k\d t^4\left(1+\frac{2\lambda}{k}\right)\right]}}\right)\right\}\nonumber\\
\end{align}
If $\lambda/k\rightarrow 0$ as $k\rightarrow\infty$, it follows as in Section \ref{subsubsec:asymp_rwmh} that
\begin{equation}
AR^{(HMC)}\asymp
\frac{(2\pi)^{k/2}}{M^k_k}\left\{1-\Phi\left(\sqrt{\frac{k}{2}}\right)\right\},
\label{eq:HMC_asymp}
\end{equation}
which is of the same asymptotic form as (\ref{eq:RWMH_asymp}), corresponding to the RWMH acceptance rate.
On the other hand, if $\lambda/k\rightarrow \infty$ as $k\rightarrow\infty$, then it follows that
\begin{equation}
AR^{(HMC)}\asymp
\frac{(2\pi)^{k/2}}{M^k_k}\left\{1-\Phi\left(\frac{\sqrt{\frac{k}{2}\left(1+\frac{\lambda}{k}\right)}}
{\sqrt{2}\sqrt{\frac{1}{M_k\d t^2}+1}}
\right)\right\},
\label{eq:HMC_asymp2}
\end{equation}
which clearly tends to zero at a much faster rate compared to (\ref{eq:HMC_asymp}).

To summarize, if $\lambda/k\rightarrow 0$ as $k\rightarrow\infty$, then both HMC and RWMH have the
same asymptotic acceptance rate, tending to zero much faster than that of additive TMCMC. On the other hand,
if $\lambda/k\rightarrow\infty$ as $k\rightarrow\infty$, the acceptance rate of HMC tends to zero much faster
than that of RWMH, while that of additive TMCMC maintains its slowest convergence rate to zero. 
Also observe that the above conclusions will continue to hold if $m_k$ and $M_k$ tend to finite positive
constants satisfying $M_k\asymp m_k$ and $k\left\vert\frac{m_k}{M_k}-1\right\vert\rightarrow 0$ as 
$k\rightarrow\infty$.


\section{Generalized Gibbs/Metropolis approaches and comparisons with TMCMC}
\label{sec:liu}


It is important to make it clear at the outset of this discussion that the goals of 
TMCMC and generalized Gibbs/Metropolis
methods are different, even though both use moves based on transformations. 
While the strength of the latter lies in improving mixing of the standard Gibbs/MH algorithms
by adding transformation-based steps to the underlying collection of usual Gibbs/MH steps, 
TMCMC is an altogether general method
of simulating from the target distribution which does not require any underlying step of Gibbs or
MH.  

The generalized Gibbs/MH methods work in the following manner.
Suppose that an underlying Gibbs or MH algorithm for exploring a target distribution
has poor mixing properties. Then in order to improve mixing, one may consider some suitable transformation
of the random variables being updated such that mixing is improved under the transformation.
Such a transformation needs to chosen carefully since it is important to ensure that 
invariance of the Markov chain is preserved under the transformation.
It is convenient to begin with an overview of the generalized Gibbs method
with a sequential updating scheme and then proceed to the discussion on the 
issues and the importance of the block updating idea in the context of improving
mixing of standard Gibbs/MH methods.

\ctn{Liu00} (see also \ctn{Liu99}) propose simulation of a transformation 
from some appropriate probability distribution, and then applying the transformation
to the co-ordinate to be updated. For example, in a $d$-dimensional
target distribution, for updating $\bx=(x_1,x_2,\ldots,x_d)$
to $\bx'=(x'_1,x_2,\ldots,x_d)$, using an additive transformation, 
one can select $\e$ from some appropriate distribution
and set $x'_1=x_1+\e$. Similarly, if a scale transformation is desired, then 
one can set $x'_1=\gamma x_1$, where $\gamma$
must be sampled from some suitable distribution. The suitable distributions of $\e$ and
$\gamma$ are chosen such that the target distribution is invariant with respect to
the move $\bx'$, the forms of which are provided in \ctn{Liu00}. 
For instance, if $\pi(\cdot)$ denotes the target distribution, then for the additive
transformation, $\e$ may be sampled from $\pi(x_1+\e,x_2,\ldots,x_d)$, and
for the multiplicative transformation, one may sample $\gamma$ 
from $|\gamma|\pi(\gamma x_1,x_2,\ldots,x_d)$.
Since direct sampling from such distributions may be impossible, \ctn{Liu00} suggest
a Metropolis-type move with respect to a transformation-invariant transition kernel.

Thus, in the generalized Gibbs method, sequentially all the variables must be updated,
unlike TMCMC, where all the variables can be updated simultaneously in a single block.
Here we note that for irreducibility issues the generalized Gibbs approach is not suitable for
updating the variables blockwise using some transformation that acts on all the variables in a given block. 
To consider a simple example, with say, $d=2$ and a single block consisting of both the variables, if
one considers the additive transformation, then starting with $\bx=(x_1,x_2)$, where $x_1<x_2$,
one can not ever reach $\bx'=(x'_1,x'_2)$, where $x'_1>0, x'_2<0$. This is because $x'_1=x_1+z$
and $x'_2=x_2+z$, for some $z$, and $x'_1>0,x'_2<0$ implies $z>-x_1$ and $z<-x_2$, which is a contradiction.
The scale transformation implies the move 
$\bx=(x_1,\ldots,x_d)\rightarrow (\gamma x_1,\ldots,\gamma x_d)=\bx'$. If
one initializes the Markov chain with all components positive,
for instance, then in every iteration, all the variables will have the same sign.
The spaces where some variables are positive and some negative will never be visited,
even if those spaces have positive (in fact, high) probabilities under the target distribution.
This shows that the Markov chain is not irreducible.
In fact, with the aforementioned approach, no transformation,
whatever distribution they are generated from, can guarantee irreducibility in general if 
blockwise updates using the transformation strategy of generalized Gibbs is used. 

Although blockwise transformations are proposed in \ctn{Liu00} (see also \ctn{Kou05} who propose
a MH-based rule for blockwise transformation), they are meant for a different purpose than that
discussed above. The strength of such blockwise transformations lies in improving the mixing behaviour
of standard Gibbs or MH algorithms.
Suppose that an underlying Gibbs or MH algorithm for exploring a target distribution
has poor mixing properties. Then in order to improve mixing, one may consider some suitable transformation
of the set of random variables being updated such that mixing is improved under the transformation.
This additional step involving transformation of the block of random variables can be obtained
by selecting a transformation from the appropriate probability distribution provided in \ctn{Liu00}.
This ``appropriate" probability distribution guarantees that stationarity of the transformed block 
of random variables is preserved. Examples reported in \ctn{Liu00}, \ctn{Muller05}, \ctn{Kou05}, etc. 
demonstrate that this 
transformation also improves the mixing behaviour of the
chain, as desired. 

Thus, to improve mixing using the methods of \ctn{Liu00} or \ctn{Kou05} one needs to run the usual Gibbs/MH 
steps, with an additional step involving transformations as discussed above. This additional step
induces more computational burden compared to the standard Gibbs/MH steps, but improved mixing
may compensate for the extra computational labour. In very high dimensions, of course, this
need not be a convenient approach since computational complexity usually makes standard Gibbs/MH approaches
infeasible. Since the additional transformation-based step works on the samples generated by
standard Gibbs/MH, impracticality of the latter implies that the extra transformation-based step of \ctn{Liu00}
for improving mixing is of little value in such cases.

It is important to point out that the generalized Gibbs/MH methods can be usefully employed by
even TMCMC to further improve its mixing properties. In other words, a step of generalized
Gibbs/MH can be added to the computational fast TMCMC. This additional step can significantly
improve the mixing properties of TMCMC. That TMCMC is much faster computationally than
standard Gibbs/MH methods imply that even in very high-dimensional situations the generalized
Gibbs/MH step can ve very much successful while working in conjunction with TMCMC.


\begin{table}[ht]
\begin{small}
 \begin{center}
\begin{tabular}{c|c|c||c|c|c} \hline
Flight no. &  Failure & Temp & Flight no. &  Failure & Temp\\ \hline
 14 & 1 & 53 & 2  & 1 &  70 \\
9  & 1 &  57 & 11  & 1 &  70 \\
23 &  1  & 58 & 6  & 0 &  72 \\
10 &  1 &  63  & 7  & 0 &  73 \\
1 &  0 &  66 & 16  & 0 &  75 \\
5 &  0 &  67 & 21  & 1 &  75 \\
13 &  0 &  67 & 19  & 0 &  76 \\
15  & 0 &  67 & 22  & 0 &  76 \\
4  & 0 &  68 & 12  & 0 &  78 \\
3  & 0 &  69 & 20  & 0 &  79 \\
8  & 0 &  70 & 18  & 0 &  81 \\
17  & 0 &  70 & & & \\ \hline
\end{tabular}
\end{center}
\caption{Challenger data. Temperature at flight time (degrees F) and failure of 
O-rings (1 stands for failure, 0 for success).}
\label{tab:challenger}
\end{small}
\end{table}

\section{Examples of TMCMC for discrete state spaces}
\label{sec:discrete_tmcmc}
The ideas developed in this paper are not confined to continuous target distributions, but also
to discrete cases. For the sake of illustration, we consider two examples below.
\begin{itemize}
\item[(i)] Consider an Ising model, where, for 
$i=1,\ldots,k$ $(k\geq 1)$, the discrete random variable $x_i$ takes the value $+1$ or $-1$
with positive probabilities. We then have
$\statesp = \{-1,1\}$. To implement TMCMC, consider the 
forward transformation 
$T(x_i,\e) = sgn(x_i+\e)$ with probability $p_i$, and choose the backward transformation as  
$T^b(x_i,\e) = sgn(x_i-\e)$ with probability $1-p_i$. Here $sgn(a)=\pm 1$ accordingly as $a>0$ or $a<0$, and $\Y = (1,\infty)$.
Note the difference with the continuous cases. Here even though neither of the transformations is 1-to-1
or onto, TMCMC works because of discreteness; the algorithm can easily be seen to satisfy detailed balance,
irreducibility and aperiodicity.
However, if $k=1$ with $x_1$ being the only variable, then, if $x_1=1$, it is possible to choose, 
with probability one, the backward move-type, yielding
$T^b(x_1,\e)=-1$. On the other hand, if $x_1=-1$, with probability one, we can choose the forward
move-type, yielding $T(x_1,\e)=1$. 
Only $2^k$ move-types are necessary
for the $k$-dimensional case for one-step irreducibility.
In discrete cases,
however, there will be no Jacobian of transformation, thereby simplifying the acceptance ratio.

\item[(ii)] For discrete state spaces like $\mathbb Z^k$, ($\mathbb Z = \{0,\pm1,\pm2,\ldots\}$) the additive transformation with single epsilon does not work. For example, with $ k =2$, if the starting state is $(1,2)$ then the chain will never reach any states $(x,y)$ where $x$ and $y$ have same parity (i.e. both even or both odd) resulting a reducible Markov chain. Thus in this case we need to have more move-types than $2^k$. For example, with some positive probability (say $r$) we may select a random coordinate and update it leaving other states unchanged. With the remaining probability (i.e. $1-r$) we may do the analogous version of the additive transformation:\\
 Let $\Y=[1,\infty)$. Then, can choose the forward transformation for each coordinate as
$T_i(x_i,\e)=x_i+[\e]$ and the backward transformation as $T^b_i(x_i,\e)=x_i-[\e]$, where $[a]$ denotes
the largest integer not exceeding $a$.

This chain is clearly ergodic and we still need only \emph{one} epsilon to update the states. 

\end{itemize}
However, in discrete cases, TMCMC reduces to Metropolis-Hastings with a mixture proposal.
But it is important to note that the implementation is much efficient and computationally cheap when TMCMC-based
methodologies developed in this paper, are used.

\end{document}